\documentclass[aps,prl,nopacs,nokeys,superscriptaddress,twocolumn,twoside,floatfix]{revtex4}

\usepackage{latexsym}
\usepackage{amsmath}
\usepackage{amssymb}
\usepackage{amsfonts}
\usepackage{ifthen}
\usepackage{revsymb}
\usepackage{yfonts}

\usepackage{graphicx,epic,eepic,epsfig,amsmath,latexsym,amssymb,verbatim,revsymb,color}

\usepackage{theorem}
\newtheorem{definition}{Definition}

\newtheorem{lemma}[definition]{Lemma}

\newtheorem{theorem}[definition]{Theorem}

\def\squareforqed{\hbox{\rlap{$\sqcap$}$\sqcup$}}
\def\qed{\ifmmode\squareforqed\else{\unskip\nobreak\hfil
\penalty50\hskip1em\null\nobreak\hfil\squareforqed
\parfillskip=0pt\finalhyphendemerits=0\endgraf}\fi}
\def\endenv{\ifmmode\;\else{\unskip\nobreak\hfil
\penalty50\hskip1em\null\nobreak\hfil\;
\parfillskip=0pt\finalhyphendemerits=0\endgraf}\fi}
\newenvironment{proof}{\noindent \textbf{{Proof~} }}{\qed}

\mathchardef\ordinarycolon\mathcode`\:
\mathcode`\:=\string"8000
\def\vcentcolon{\mathrel{\mathop\ordinarycolon}}
\begingroup \catcode`\:=\active
  \lowercase{\endgroup
  \let :\vcentcolon
  }

\newcommand{\nc}{\newcommand}
\nc{\rnc}{\renewcommand}
\nc{\beq}{\begin{equation}}
\nc{\eeq}{{\end{equation}}}
\nc{\beqa}{\begin{eqnarray}}
\nc{\eeqa}{\end{eqnarray}}
\nc{\lbar}[1]{\overline{#1}}
\nc{\bra}[1]{\langle#1|}
\nc{\ket}[1]{|#1\rangle}
\nc{\ketbra}[2]{|#1\rangle\!\langle#2|}
\nc{\braket}[2]{\langle#1|#2\rangle}
\nc{\proj}[1]{| #1\rangle\!\langle #1 |}
\nc{\avg}[1]{\langle#1\rangle}
\rnc{\max}{\operatorname{max}}
\nc{\Rank}{\operatorname{Rank}}
\nc{\smfrac}[2]{\mbox{$\frac{#1}{#2}$}}
\nc{\Tr}{\operatorname{Tr}}
\nc{\ox}{\otimes}
\nc{\dg}{\dagger}
\nc{\dn}{\downarrow}
\nc{\cA}{{\cal A}}
\nc{\cB}{{\cal B}}
\nc{\cC}{{\cal C}}
\nc{\cD}{{\cal D}}
\nc{\cE}{{\cal E}}
\nc{\cF}{{\cal F}}
\nc{\cG}{{\cal G}}
\nc{\cH}{{\cal H}}
\nc{\cI}{{\cal I}}
\nc{\cJ}{{\cal J}}
\nc{\cK}{{\cal K}}
\nc{\cL}{{\cal L}}
\nc{\cM}{{\cal M}}
\nc{\cN}{{\cal N}}
\nc{\cO}{{\cal O}}
\nc{\cP}{{\cal P}}
\nc{\cR}{{\cal R}}
\nc{\cS}{{\cal S}}
\nc{\cT}{{\cal T}}
\nc{\cX}{{\cal X}}
\nc{\cZ}{{\cal Z}}
\nc{\csupp}{{\operatorname{csupp}}}
\nc{\qsupp}{{\operatorname{qsupp}}}
\nc{\var}{\operatorname{var}}
\nc{\rar}{\rightarrow}
\nc{\lrar}{\longrightarrow}
\nc{\polylog}{\operatorname{polylog}}
\nc{\1}{{\openone}}
\nc{\id}{{\openone}}

\nc{\RR}{{{\mathbb R}}}
\nc{\CC}{{{\mathbb C}}}
\nc{\FF}{{{\mathbb F}}}
\nc{\NN}{{{\mathbb N}}}
\nc{\ZZ}{{{\mathbb Z}}}
\nc{\PP}{{{\mathbb P}}}
\nc{\QQ}{{{\mathbb Q}}}
\nc{\UU}{{{\mathbb U}}}
\nc{\EE}{{{\mathbb E}}}

\nc{\be}{\begin{equation}}
\nc{\ee}{{\end{equation}}}
\nc{\bea}{\begin{eqnarray}}
\nc{\eea}{\end{eqnarray}}
\nc{\<}{\langle}
\nc{\Hom}[2]{\mbox{Hom}(\CC^{#1},\CC^{#2})}
\nc{\rU}{\mbox{U}}

\nc{\ob}[1]{#1}

\newcommand{\setA}{\mathcal{A}}
\newcommand{\setB}{\mathcal{B}}
\newcommand{\setT}{\mathcal{T}}
\newcommand{\Real}{\mathbb{R}}
\newcommand{\hil}{\mathcal{H}}

\newcommand{\inp}[2]{\langle{#1}|{#2}\rangle}

\newcommand{\cancel}[1]{}

\begin{document}

\title{Higher entropic uncertainty relations for anti-commuting observables}

\author{Stephanie Wehner}
\affiliation{Centrum voor Wiskunde en Informatica, Kruislaan 413, 1098 SJ Amsterdam, The Netherlands}
\email{s.d.c.wehner@cwi.nl}

\author{Andreas Winter}
\affiliation{Department of Mathematics, University of Bristol, Bristol BS8 1TW, U.K.}
\affiliation{Quantum Information Technology Lab, National University of Singapore,
 2 Science Drive 3, Singapore 117542}
\email{a.j.winter@bris.ac.uk}

\date{3 October 2007}

\maketitle

\noindent
{\bf Uncertainty relations
provide one of the most powerful formulations
of the quantum mechanical principle of complementarity. 
Yet, very little is known about such uncertainty relations for more than two measurements. 
Here, we show that sufficient unbiasedness for a set of
binary observables, in the sense of mutual anti-commutation,
is good enough to obtain maximally strong uncertainty relations in terms
of
the Shannon entropy.
We also prove nearly optimal relations for the collision entropy. This is the first systematic and explicit approach to finding an arbitrary number of measurements for which we obtain maximally strong uncertainty relations.
Our results have immediate applications to quantum cryptography. 
}

\bigskip\noindent
Uncertainty relations lie at the very core of quantum mechanics. 
For any observable, it only has sharp values (in the sense that the measurement
outcome is deterministic) for its own eigenstates.
However, for any other state, the distribution
of measurement outcomes is more or less smeared out, or more conveniently
expressed: its entropy is is strictly positive. 
Hence, if two or more observables have no eigenstates in common,
the sum of these respective entropies is strictly greater than $0$ for any state we may measure. We thereby say that a set of observables is more ``incompatible'' 
than another, if this sum takes on a larger value. But what makes observables more ``incompatible''?
Or rather, what characterizes maximally ``incompatible'' observables? Here, we show how to obtain maximally
strong uncertainty relations for a large number of binary observables that exhibit simple geometrical properties.

Uncertainty relations are most well-known in the form proposed by
Heisenberg~\cite{heisenberg:uncertainty} and generalized by Robertson~\cite{robertson:uncertainty}. Entropic uncertainty
relations are an alternative way to state Heisenberg's uncertainty principle. They are frequently a
more useful characterization, because the ``uncertainty'' is lower bounded by a quantity that only 
depends on the eigenstates of the observables, and not on the actual physical quantity to be
measured~\cite{bm:uncertainty,deutsch:uncertainty}, as in Heisenberg's formulation with standard
deviations -- see also the more recent paper~\cite{GII}.
Following a conjecture by Kraus~\cite{kraus:entropy}, Maassen and Uffink~\cite{maassen:entropy}
proved an entropic uncertainty relation for \emph{two} observables.
In particular, they showed that if we measure any state $\rho \in \hil$ with $\dim \hil = d$ 
using observables with eigenbases $\setA = \{\ket{a_1},\ldots,\ket{a_d}\}$ and 
$\setB = \{\ket{b_1},\ldots,\ket{b_d}\}$
respectively, we have
$$
\frac{1}{2}\bigl( H(\setA|\rho) + H(\setB|\rho) \bigr) \geq - \log c(\setA,\setB),
$$
where $c(\setA,\setB) = \max\{|\inp{a}{b}| : \ket{a} \in \setA, \ket{b} \in \setB\}$ and 
$H(\setA|\rho) = - \sum_{i=1}^d \bra{a_i}\rho\ket{a_i} \log \bra{a_i}\rho\ket{a_i}$
is the Shannon entropy
arising from measuring the state $\rho$ in basis $\setA$. Here, the most ``incompatible'' 
measurements arise from choosing $\setA$ and $\setB$ to be
\emph{mutually unbiased bases} (MUB). That is,
for any $\ket{a} \in \setA$ and any $\ket{b} \in \setB$ we have $|\inp{a}{b}| = 1/\sqrt{d}$, giving
us a lower bound of $\frac{1}{2}\log d$.
Clearly, this bound is tight: Choosing $\rho = \proj{a_i}$ for $\ket{a_i} \in \setA$ gives
us exactly $\frac{1}{2}\log d$, with maximum uncertainty
for one of the two observables and none for the other.

But how about more than two observables? Sadly, very little is known about this case so far. Yet, this question
not only eludes our current understanding of quantum mechanics, but also has practical consequences
for quantum cryptography in the bounded storage model, where proving the security of protocols ultimately
reduces to finding such relations~\cite{serge:new}. Proving new entropic uncertainty relations
could thus give rise to new protocols. Furthermore, uncertainty relations for more than two
measurements could also be useful to understand other quantum effects 
that are derived from such relations,
such as
locking classical information in quantum states~\cite{terhal:locking}.
Sanchez-Ruiz~\cite{sanchez:entropy,sanchez:entropy2,sanchez:entropyD2} 
has shown that for a full set of $d+1$ MUBs $\setA_1,\ldots,\setA_{d+1}$, we have
$$
\frac{1}{d+1} \sum_{j=1}^{d+1} H(\setA_{j}|\rho) \geq  \log\left(\frac{d+1}{2}\right),
$$
and for $d=2$ gave a lower bound of $2/3$.
Indeed, strong uncertainty relations for a smaller number of bases do exist.
If we choose a set $\setT$ of
$(\log d)^4$ bases uniformly at random, then (with high probability)
we have that for all states $\rho$:
$\frac{1}{|\setT|} \sum_{\setB \in \setT} H(\setB|\rho)
  \geq \log d -3$~\cite{winter:randomizing}. 
This means that
there exist $(\log d)^4$ bases for which the sum of entropies is very large, i.e., measurements in such bases
are very incompatible. However, no explicit constructions are known.
It may be tempting to conjecture that simply choosing our measurements to be mutually unbiased
leads to strong uncertainty relations in general. In fact, when choosing bases at random they will be almost mutually unbiased. 
In this case, we might
expect the entropy average to be quite large: if the state to be measured is an eigenstate of one of the
bases, the corresponding entropy average will be $\left( 1- \frac{1}{|\setT|} \right)\log d$.
This value is thus clearly an upper bound on the minimum entropy average
$\min_{\rho} \frac{1}{|\setT|} \sum_{\setB \in \setT} H(\setB|\rho)$ for any
set of bases, mutually unbiased or not.
Perhaps surprisingly, however, choosing the bases to be mutually unbiased is not the right property:
there exists up to $|\setT| \leq \sqrt{d}$ mutually unbiased bases for which
$\min_{\rho} \frac{1}{|\setT|} \sum_{\setB \in \setT} H(\setB|\rho) = \frac{1}{2}\log d$~\cite{wehner06c}.
Note that the right hand side is a lower bound for any set of MUBs, since it is the average of pairs
of entropies to which we can apply the uncertainty relation by Maassen and Uffink~\cite{maassen:entropy}. 
Hence we call this the trivial lower bound. 
When considering entropic uncertainty relations as a measure of ``incompatibility'', we must thus 
look for different properties to obtain strong uncertainty relations. But, what properties
lead to strong entropic uncertainty relations for more than two observables?

Here, we show that for binary observables we obtain maximally strong uncertainty relations
for the Shannon entropy if they satisfy the property that they \emph{anti-commute}. We also obtain
a nearly optimal uncertainty relation for the collision entropy (R{\'e}nyi entropy of order $2$) 
$H_2(X) = - \log \sum_x P_X(x)^2$ that is of particular relevance to cryptography. As we will see,
we can take the anti-commuting observables to have a particularly simple form that in principle allows
us to apply our result to quantum cryptography using present-day technology.

\section{Clifford algebra}
For our result we will make use of the structure of
Clifford algebra~\cite{lounesto:book,doran:lasenby:book,dietz:blochsphere},
which has many beautiful geometrical properties of which we shall use a few.
For any integer $n$, the free real associative algebra generated by
$\Gamma_1,\ldots,\Gamma_{2n}$, subject to the anti-commutation
relations
\begin{equation}
  \label{eq:anti}
  \{ \Gamma_j,\Gamma_k \} = \Gamma_j\Gamma_k + \Gamma_k\Gamma_j = 2\delta_{jk} \1,
\end{equation}
is called \emph{Clifford algebra}. We briefly recall its most
essential properties that we will use in this text. 
The Clifford algebra has a unique representation by
Hermitian matrices on $n$ qubits (up to unitary equivalence) which we fix henceforth.
This representation can be obtained via the famous Jordan-Wigner transformation~\cite{JordanWigner}:
\begin{align*}
  \Gamma_{2j-1} &= Z^{\ox(j-1)} \ox X \ox \1^{\ox(n-j)}, \\
  \Gamma_{2j}   &= Z^{\ox(j-1)} \ox Y \ox \1^{\ox(n-j)},
\end{align*}
for $j=1,\ldots,n$, where we use $X$, $Y$ and $Z$ to denote the Pauli matrices.

Let us first consider these operators themselves. Evidently, 
each operator $\Gamma_i$ has exactly two eigenvalues $\pm 1$:
Let $\ket{\eta}$ be an eigenvector of $\Gamma_i$ with eigenvalue $\lambda$.
From $\Gamma_i^2 = \1$ we have that $\lambda^2 = 1$. Furthermore,
we have $\Gamma_i (\Gamma_j \ket{\eta}) = - \lambda \Gamma_j \ket{\eta}$.
We can therefore express each $\Gamma_i$ as
$$
\Gamma_i = \Gamma_i^0 - \Gamma_i^1,
$$
where $\Gamma_i^0$ and $\Gamma_i^1$ are projectors onto the positive and
negative eigenspace of $\Gamma_i$ respectively. Furthermore, note that
we have for $i\neq j$
$$
\Tr(\Gamma_i \Gamma_j) = \frac{1}{2} \Tr(\Gamma_i \Gamma_j + \Gamma_j \Gamma_i) = 0.
$$
That is, all such operators are orthogonal.
Hence, the positive and negative eigenspaces of such operators
are similarly mutually unbiased than bases can be: we have that
for all $i \neq j$ 
$$
\Tr(\Gamma_i \Gamma_j^0) = \Tr(\Gamma_i \Gamma_j^1).
$$

The crucial aspect of the Clifford algebra that makes it so useful
in geometry is that we can view the operators
$\Gamma_1,\ldots,\Gamma_{2n}$ as $2n$ orthogonal vectors forming a basis
for $\Real^{2n}$. Each vector $a = (a_1,\ldots,a_{2n}) \in \Real^{2n}$ can
then be written as $a = \sum_j a_j \Gamma_j$. Note that the inner product of two
vectors obeys $a \cdot b = \sum_j a_j b_j \1 = \{a,b\}/2$, where $ab$ is the Clifford product which
here is just equal to the matrix product. Hence, anti-commutation takes a geometric
meaning within the algebra: two vectors anti-commute if and only if they are
orthogonal.
Evidently, if we now transform the generating set
of $\Gamma_j$ linearly to obtain the new operators
\[
  \Gamma_k' = \sum_{j} T_{jk}\Gamma_j,
\]
then the set $\{ \Gamma_1',\ldots,\Gamma_{2n}'\}$ satisfies the anti-commutation
relations iff $(T_{jk})_{jk}$ is an orthogonal matrix: these are exactly the operations
which preserve the inner product. Because of the uniqueness of representation,
there exists a matching unitary $U(T)$ of ${\cal H}$ which transforms the
operator basis on the Hilbert space level, by conjugation:
\[
  \Gamma_j' = U(T) \Gamma_j U(T)^\dagger.
\]
Essentially, we can think of the positive and negative eigenspace of such operators 
as the positive and negative direction of the basis vectors. 
We can visualize the $2n$ basis
vectors with the help of a $2n$-dimensional hypercube. 
Each basis vector determines two opposing faces of the hypercube, 
where we can think of the two faces as corresponding to the positive and negative eigenspace of each operator. 
\begin{figure}[ht]
\includegraphics[scale=0.7]{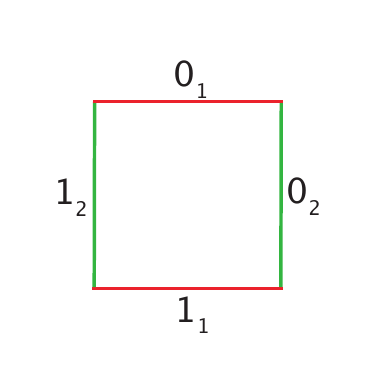}
\caption{$2$-cube, corresponding to $n=1$.}
\end{figure}
Note that the face of an $2n$-dimensional hypercube is a $2n-1$ dimensional hypercube itself.
\begin{figure}[ht]
\includegraphics[scale=0.6]{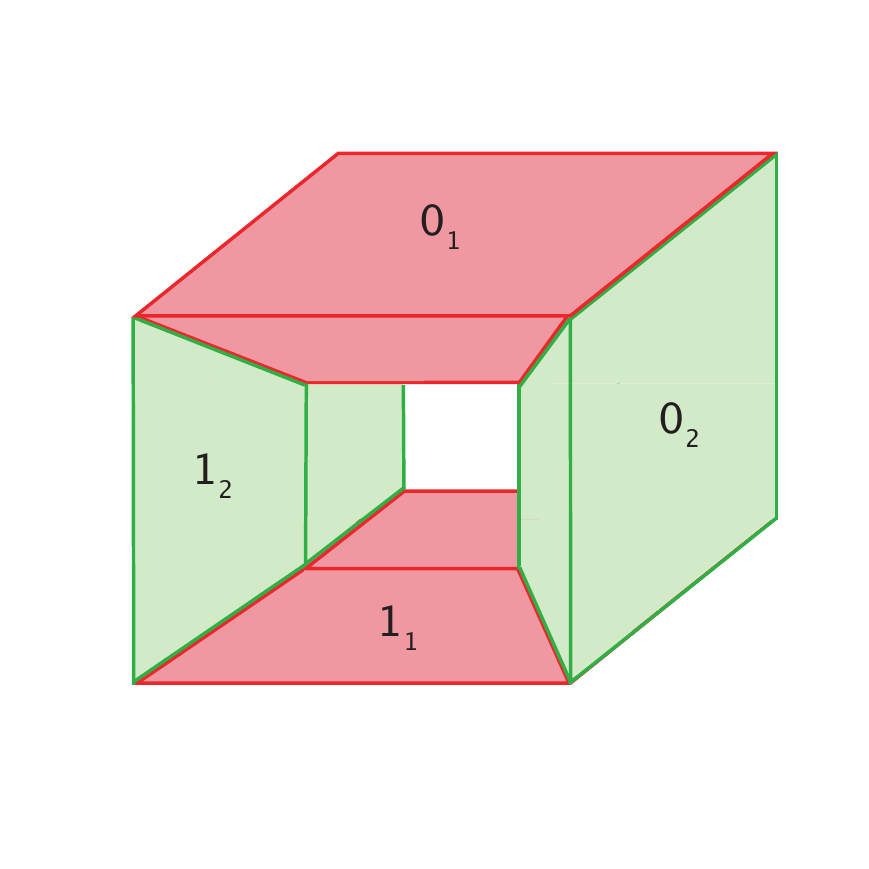}
\caption{$4$-cube, corresponding to $n=2$.}
\end{figure}

It will be particularly useful that the collection of operators
\[\begin{split}
  \1           & \\
  \Gamma_j     & \phantom{===} (1\leq j\leq 2n) \\
  \Gamma_{jk}  &= i\Gamma_j\Gamma_k \ (1\leq j < k \leq 2n) \\
  \Gamma_{jk\ell} &= \Gamma_j\Gamma_k\Gamma_\ell \ (1\leq j < k < \ell \leq 2n) \\
    \vdots     &\\
  \Gamma_{12\ldots (2n)} &= i\Gamma_1\Gamma_2 \cdots \Gamma_{2n} =: \Gamma_0
\end{split}\]
form an orthogonal basis for the $d \times d$ complex matrices for $d=2^n$, again
by the anti-commutation relations. By counting, the above operators
form a complete operator basis with respect to the Hilbert-Schmidt
inner product. Notice that the products with an odd number of factors are
Hermitian, while the ones with an even number of factors are skew-Hermitian,
so in the definition of the above operators we introduce a factor of
$i$ to all with an even number of indices to make the whole set a
real basis for the Hermitian operators. Working out the above terms using
the representation from above, we can see that this gives us the familiar
Pauli basis consisting of elements $B_j^1 \otimes \ldots \otimes B_j^n$
with $B_j^i \in \{\1,X,Y,Z\}$.

Hence we can write every state $\rho$ on ${\cal H}$ as
\begin{equation}
  \label{eq:op-basis}
  \rho = \frac{1}{d}\left( \1 + \sum_j g_j\Gamma_j + \sum_{j<k} g_{jk}\Gamma_{jk}
                              + \ldots + g_0 \Gamma_0 \right).
\end{equation}
This expansion has been used before in quantum information theory,
see e.g.~\cite{dietz:blochsphere}.
The (real valued) coefficients $(g_1,\ldots,g_{2n})$ in this
expansion are called ``vector'' components,
the ones belonging to degree $k > 1$ products of $\Gamma$'s are
``tensor'' or $k$-vector components. $k$-vectors also have very nice geometric
interpretation within the algebra: they represent oriented plane and higher
volume elements.
The -- unique -- coordinate $\Gamma_0$ of degree $2n$ also plays special role
(it corresponds to the volume element in $\Real^{2n}$), and is called the
``pseudo-scalar'' component. Note that it anti-commutes with all the $\Gamma_j$,
which has another important consequence:
Substituting $\Gamma_0$ for any of the $\Gamma_j$ again yields
a generating set of the Clifford algebra, hence there exists a unitary on
${\cal H}$ taking the original to the new basis by conjugation.

The vector and pseudo-scalar components of the Clifford algebra span a
$(2n+1)$-dimensional space isomorphic to $\Real^{2n+1}$: indeed,
extending the $\text{O}(2n)$ symmetry of $\text{span}\{\Gamma_1,\ldots,\Gamma_{2n}\}$,
the extended $\text{span}\{\Gamma_0,\Gamma_1,\ldots,\Gamma_{2n}\}$ has
the symmetry of $\text{SO}(2n+1)$: for every special-orthogonal
$(2n+1) \times (2n+1)$ matrix $\tilde{T}$, we can write transformed Clifford operators
$\Gamma_k' = \sum_{j=0}^{2n} \tilde{T}_{jk} \Gamma_j$ obeying the anti-commutation
relations. As before (but now this requires an additional proof that we provide in the appendix using
the condition $\det \tilde{T} = 1$), there exists a unitary $U(\tilde{T})$ of the
underlying Hilbert space ${\cal H}$ such that for all $j=0,\ldots,2n$,
$\Gamma_j' = U(\tilde{T}) \Gamma_j U(\tilde{T})^\dagger$.

Using the orthogonal group symmetry of the Clifford algebra,
we show the following lemma in the appendix.
\begin{lemma}\label{reduceGrades}
The linear map $\mathbb{P}$ taking $\rho$ as
in eq.~(\ref{eq:op-basis}) to
\begin{equation}
  \label{eq:projection}
  \mathbb{P}(\rho) := \frac{1}{d}\left( \1 + \sum_{j=0}^{2n} g_j\Gamma_j \right)
\end{equation}
is positive. I.e., if $\rho$ is a state, then so is $\mathbb{P}(\rho)$,
and in this case $\sum_{j=0}^{2n} g_j^2 \leq 1$.
Conversely, if $\sum_{j=0}^{2n} g_j^2 \leq 1$, then
\[
  \sigma = \frac{1}{d}\left( \1 + \sum_{j=0}^{2n} g_j\Gamma_j \right)
\]
is positive semidefinite, hence a state.
\end{lemma}

\noindent
It is interesting to note that the map $\mathbb{P}$ is positive, 
but not \emph{completely positive}, for any $n>1$, as one can see
straightforwardly by looking at it's Choi-Jamio\l{}kowski operator.

\section{Applications}
We now first use the tools from above to prove a ``meta''-uncertainty relation, from which
we will then derive two new entropic uncertainty relations.
Evidently, we have immediately from the above that
\begin{lemma}\label{metaUR}
Let $\rho \in \hil$ with $\dim\hil = 2^n$ be a quantum state,
and consider $K \leq 2n+1$ anti-commuting observables $\Gamma_j$ as defined above.
Then,
$$
  \sum_{j=0}^{K-1} \bigl( \Tr(\rho\Gamma_j) \bigr)^2
     \leq \sum_{j=0}^{2n} \bigl( \Tr(\rho\Gamma_j) \bigr)^2
     \leq 1.
$$
\qed
\end{lemma}
Our result is essentially a generalization of the Bloch sphere picture to higher dimensions
(see also~\cite{dietz:blochsphere}):
For $n=1$ ($d=2$) the state is parametrized by 
$\rho = \frac{1}{2}(\1 + g_1 \Gamma_1 + g_2 \Gamma_2 + g_0 \Gamma_0)$
where $\Gamma_1 = X$, $\Gamma_2 = Z$ and $\Gamma_0 = Y$ are the familiar Pauli matrices. 
Lemma~\ref{metaUR} tells
us that $g_0^2 + g_1^2 + g_2^2 \leq 1$, i.e., the state must lie inside the Bloch sphere.
Our result may be of independent interest, since it is often hard to find conditions on the coefficients
$g_1,g_2,\ldots$ such that $\rho$ is a state.

Notice that the $g_j = \Tr(\rho\Gamma_j)$ are directly interpreted
as the expectations of the observables $\Gamma_j$. Indeed, $g_j$ is
precisely the bias of the $\pm1$-variable $\Gamma_j$:
\[
  \Pr\{ \Gamma_j = 1 | \rho \} = \frac{1+g_j}{2}.
\]
Hence, we can interpret Lemma~\ref{metaUR} as a form of uncertainty relation
between the observables $\Gamma_j$: if one or more of the observables have a large
bias (i.e., they are more precisely defined), this limits the bias of
the other oberservables (i.e., they are closer to uniformly distributed).

Indeed, Lemma~\ref{metaUR} has strong consequences for the R\'{e}nyi and von Neumann
entropic averages
$$
  \frac{1}{K} \sum_{j=0}^{K-1} H_\alpha\left(\Gamma_j|\rho\right),
$$
where $H_\alpha(\Gamma_j|\rho)$ is the R{\'e}nyi entropy at $\alpha$
of the probability distribution arising from measuring the state $\rho$ with
observable $\Gamma_j$. 
The minima of such expressions
can be interpreted as giving entropic uncertainty relations, as we shall now do
for $\alpha=2$ (the collision entropy) and $\alpha=1$ (the Shannon entropy).

\begin{theorem}
Let $\dim\hil = 2^n$,
and consider $K \leq 2n+1$ anti-commuting observables as defined above.
Then,
$$
\min_{\rho} \frac{1}{K} \sum_{j=0}^{K-1} H_2\left(\Gamma_j|\rho\right)
      = 1 - \log\left( 1+\frac{1}{K} \right)
      \sim 1 - \frac{\log e}{K},
$$
where $H_2(\Gamma_j|\rho) = - \log \sum_{b \in \{0,1\}} \Tr(\Gamma_j^b \rho)^2$,
and the minimization is taken over all states $\rho$.
The latter holds asymptotically for large $K$.
\end{theorem}
\begin{proof}
Using the fact that $\Gamma_j^b = (\1 + (-1)^b \Gamma_j)/2$ we
can first rewrite
\begin{equation*}\begin{split}
  \frac{1}{K} \sum_{j=0}^{K-1} H_2\left(\Gamma_j|\rho\right)
      &=    - \frac{1}{K} \sum_{j=0}^{K-1} \log \left[\frac{1}{2}\left(1 + \Tr(\rho\Gamma_j)^2\right)\right]  \\
      &\geq - \log \left( \frac{1}{2K} \sum_{j=0}^{K-1} \left(1 + g_j^2\right) \right) \\
      &\geq 1 - \log\left(1 + \frac{1}{K}\right),
\end{split}\end{equation*}
where the first inequality follows from Jensen's inequality and the concavity of the log,
and the second from Lemma~\ref{metaUR}.
Clearly, the minimum is attained if
all $g_j = \Tr(\rho\Gamma_j) = \sqrt{\frac{1}{K}}$.
It follows from Lemma~\ref{reduceGrades} that our inequality is tight.
Via the Taylor expansion of $\log\left(1 + \frac{1}{K}\right)$ we obtain
the asymptotic result for large $K$.
\end{proof}

For the Shannon entropy ($\alpha=1$) we obtain something even nicer:
\begin{theorem}
Let $\dim\hil = 2^n$,
and consider $K \leq 2n+1$ anti-commuting observables as defined above.
Then,
$$
\min_{\rho} \frac{1}{K} \sum_{j=0}^{K-1} H(\Gamma_j|\rho) = 1 - \frac{1}{K},
$$
where $H(\Gamma_j|\rho) = - \sum_{b\in \{0,1\}} \Tr(\Gamma_j^b \rho) \log \Tr(\Gamma_j^b \rho)$,
and the minimization is taken over all states $\rho$.
\end{theorem}
\begin{proof}
To see this, note that by rewriting our objective as above,
we observe that we need to minimize the expression
\[
  \frac{1}{K} \sum_{j=0}^{K-1} H\left( \frac{1 \pm \sqrt{t_j}}{2} \right),
\]
subject to $\sum_j t_j \leq 1$ and $t_j \geq 0$, via the identification
$t_j = (\Tr(\rho\Gamma_j))^2$. An elementary calculation (included in the
appendix for completeness) shows that
the function $f(t) = H\left( \frac{1 \pm \sqrt{t}}{2} \right)$ is
concave in $t\in[0;1]$. Hence, by Jensen's inequality (read in the
opposite direction), the minimum
is attained with all the $t_j$ being extremal, i.e.~one of the $t_j$
is $1$ and the others are $0$,
giving just the lower bound of $1-\frac{1}{K}$.
\end{proof}

\medskip
It is clear that based on Lemma~\ref{reduceGrades} one can
derive similar uncertainty relations for other R\'{e}nyi entropies
($\alpha \neq 1,2$) by performing the analogous optimization. 
We stuck to the two values above as they are
the most relevant in view of the existing literature; for example,
using the same convexity arguments as for $\alpha=2$, we obtain
for $\alpha = \infty$,
\[
  \frac{1}{K} \sum_{j=0}^{K-1} H_\infty\left(\Gamma_j|\rho\right)  \geq 1 - \log\left( 1+\frac{1}{\sqrt{K}} \right).
\]
This should be compared to Deutsch's inequality~\cite{deutsch:uncertainty}
for the case of two mutually unbiased
bases of a qubit, because the latter really is about $H_\infty$.

\section{Discussion}
We have shown that anti-commuting Clifford observables obey the
strongest possible uncertainty relation for the von Neumann entropy.
It is interesting that in the process of the proof, however,
we have found three uncertainty type inequalities
(the sum of squares bound, the bound on $H_2$, and finally the bound on
$H_1$), and all three have a different structure of attaining the
limit. The sum of squares bound can be achieved in every direction
(meaning for every tuple satisfying the bound we get one attaining
it by multiplying all components by some appropriate factor),
the $H_2$ expression requires all components to be equal,
while the $H_1$ expression demands exactly the opposite.

Our result for the collision entropy is slightly suboptimal but
strong enough for all cryptographic purposes. Indeed, one could use our
entropic uncertainty relation
in the bounded quantum storage setting to construct, for instance, 
$1$-out-of-$K$ oblivious transfer protocols
analogous to~\cite{serge:new}. Here, instead of encoding a single 
bit into either the computational or Hadamard
basis, which gives us a 1-out-of-2 oblivious transfer, we now encode a single bit into the 
positive or negative eigenspace of each
of these $K$ operators. It is clear from the representation of such 
operators discussed earlier, that such
an encoding can be done experimentally as easily as encoding a single
bit into three mutually unbiased basis
given by the Pauli operators $X$, $Y$ and $Z$. Indeed, our construction 
can be seen as a direct extension of such
an encoding: we obtain the uncertainty relations for these three MUBs used in~\cite{serge:new},
previously proved by Sanchez-Ruiz~\cite{sanchez:entropy,sanchez:entropy2},
as a special case of our analysis for $K=3$ ($d=2$).

Alas, strong uncertainty relations for measurements with more than two outcomes
remain inaccessible to us.
It has been shown~\cite{serge:personal} that uncertainty 
relations for more outcomes can be obtained via
a coding argument from uncertainty relations as we construct 
them here. Yet, these seem far from optimal.
A natural choice would be to consider the generators of a generalized
Clifford algebra~\cite{Morris:gen-clifford-1,Morris:gen-clifford-2},
yet this algebra does not have the nice symmetry properties which 
enabled us to implement operations on the
vector components above. It remains an exciting open question, 
whether such operators form a good
generalization, or whether we must continue our search for new properties.

\bigskip\noindent
{\bf Acknowledgments.}
The authors acknowledge support by the EC project ``QAP'' (IST-2005-015848).
SW was additionally supported by the NWO vici project 2004-2009.
AW was additionally supported by the U.K.~EPSRC via the ``IRC QIP''
and an Advanced Research Fellowship.
SW thanks Andrew Doherty for an explanation of the Jordan-Wigner transform.

\appendix
\section{Appendix}

\noindent
{\bf SO(2n+1) structure.} While the orthogonal group symmetry of the
``vector'' component of the Clifford algebra, spanned by the generators
$\{\Gamma_1,\ldots,\Gamma_{2n}\}$, is usually covered in textbook
accounts, the symmetry of the extended set
$\{\Gamma_0,\Gamma_1,\ldots,\Gamma_{2n}\}$, including the pseudo-scalar
element, seems much less well-known. It is quite natural to consider this
set as all its elements mutually anti-commute, so any family
${\cal K} = (k_1,\ldots, k_{2n})$ of $2n$ pairwise distinct
elements will generate the full Clifford algebra. Hence there exists
a unitary $U({\cal K})$ mapping the original generators $\Gamma_j$
to the $\Gamma_{k_j}$:
\[
  \Gamma_{k_j} = U({\cal K}) \Gamma_j U({\cal K})^\dagger.
\]

The initial observation is that indeed an orthogonal transformation $T$
of the $2n$ generators extends to a special-orthogonal transformation
$\tilde{T} = (\det T) \oplus T$ of the extended set, since
\[
  \Gamma_0' = U(T) \Gamma_0 U(T)^\dagger 
            = i \Gamma_1' \cdots \Gamma_{2n}'
            = (\det T) \Gamma_0.
\]
A nice and easy geometrical way of seeing this is via the higher-dimensional
analogue of the well-known Euler angle parametrisation of orthogonal
matrices (see~\cite{goldstein:mechanics}):

\medskip\noindent
{\it Euler Angle Decomposition~\cite{HoffmanRaffenettiRuedenberg:euler-angles}.}
\emph{%
  Let $T$ be an $N\times N$ orthogonal matrix. Then there exist angles
  $\theta_{jk} \in [0;2\pi)$ for $1\leq j < k \leq N$, such that
  \[
    T = E_1^{\det T} \prod_{j<k} R_{jk}(\theta_{jk}),
  \]
  where $E_1^\epsilon = \epsilon \proj{1} + \sum_{i>1} \proj{j}$ is
  either the identity or the reflection along the first coordinate axis,
  and $R_{jk}(\theta)$ is the rotation by angle $\theta$ in the plane
  spanned by the $j$th and $k$th coordinate axes, i.e.
  \[\begin{split}
    R_{jk}(\theta) &= \cos\theta \proj{j} + \sin\theta\ket{k}\bra{j} \\
                   &\phantom{=}
                      - \sin\theta\ket{j}\bra{k} + \cos\theta\proj{k} 
                    + \sum_{i\neq j,k} \proj{i}.
  \end{split}\]
  (The product is to be taken in some fixed order of the indices, say
  lexicographically.)
  \qed %
}

\medskip
With this, we only have to understand how $\Gamma_0$ transforms under the action
of the elementary transformations $E_1^\epsilon$ and $R_{jk}(\theta)$. Clearly,
under the former,
\[
  \Gamma_0' = \epsilon \Gamma_0,
\]
while for the latter (using the abbreviations $c=\cos\theta$ and $s=\sin\theta$),
\[\begin{split}
  \Gamma_0' &= i\Gamma_1' \cdots \Gamma_{2n}'                \\
            &= i\Gamma_1 \cdots \Gamma_{j-1} \cdot               \\
            &\phantom{===}
               (c\Gamma_j + s\Gamma_k) \Gamma_{j+1} \cdots \Gamma_{k-1} (-s\Gamma_j + c\Gamma_k) \cdot \\
            &\phantom{======}
                \Gamma_{k+1} \cdots \Gamma_{2n}                                                  \\
            &= i(c^2+s^2)\Gamma_1 \cdots \Gamma_{2n}         \\
            &\phantom{=}
               + i(-cs+sc) \Gamma_1 \cdots \Gamma_{j-1}
                           \Gamma_{j+1} \cdots \Gamma_{k-1}
                           \Gamma_{k+1} \cdots \Gamma_{2n}   \\
            &= \Gamma_0.
\end{split}\]

Now, for a general special-orthogonal transformation $\tilde{T}$ of the
$2n+1$ coordinates of the extended set, the Euler angle decomposition
gives
\[
  \tilde{T} = \prod_{0\leq j<k \leq 2n} R_{jk}(\theta_{jk}).
\]
Then, the unitary representation $U(\tilde{T})$ clearly has to be the product
of terms $U\bigl(R_{jk}(\theta)\bigr)$. For $1\leq j < k \leq 2n$ we know
already what these are, as the transformation is only one of the
generating set $\{\Gamma_1,\ldots,\Gamma_{2n}\}$ (and by the above observation
the pseudo-scalar $\Gamma_0$ is indeed left alone, as required);
for $0=j < k \leq 2n$ on the other hand, we first map the generating
set ${\cal K}=\{\Gamma_0,\Gamma_k.\ldots\}$ to $\{\Gamma_1,\ldots,\Gamma_{2n}\}$
by the unitary $U({\cal K})^\dagger$, then apply the unitary belonging to $R_{12}(\theta)$
and then map the generators back via $U({\cal K})$. This clearly
implements
\[
  U\bigl(R_{jk}(\theta)\bigr) = U({\cal K}) U\bigl(R_{12}(\theta)\bigr)  U({\cal K})^\dagger,
\]
and we are done.
\qed

\bigskip\noindent
{\bf Proof of Lemma 1.}
First, we show that there exists a unitary $U$ such that
$\rho' = U\rho U^\dagger$ has no pseudo-scalar, and only one nonzero vector component,
say at $\Gamma_1$, which we can choose to be $g_1' = \sqrt{\sum_{j=0}^{2n} g_j^2}$.
Indeed, there is a special-orthogonal transformation $T^{-1}$ of the
coefficient vector $(g_0,g_1,\ldots,g_{2n})$ to a vector whose zeroeth
as well as second till last components are all $0$: since the length is preserved, this
is consistent with the first component becoming $\sqrt{\sum_j g_j^2}$.

Now, let $U = U(T)$ be the corresponding unitary of the Hilbert space.
By the above-mentioned representation of $\text{SO}(2n+1)$ on $\hil$,
we arrive at a new, simpler looking state
\[\begin{split}
  \rho' &= U(T) \rho U(T)^\dagger \\
        &= \frac{1}{d}\left( \1 + g_1'\Gamma_1 + \sum_{j<k} g'_{jk}\Gamma_{jk}
                                                   + \ldots + 0\,\Gamma_0 \right),
\end{split}\]
for some $g'_{jk}$, etc.

There exist of course orthogonal transformations $F_j$ that take
$\Gamma_k$ to $(-1)^{\delta_{jk}}\Gamma_k$. Such transformations
flip the sign of a chosen Clifford generator. They can be extended to a special
orthogonal transformation of $\text{span}\{\Gamma_0,\ldots,\Gamma_{2n}\}$ by
also flipping the sign of $\Gamma_0$: $F_j\Gamma_0 = -\Gamma_0$.
(Using the geometry of the Clifford algebra it is easy to see
that $U(F_j) = \Gamma_0 \Gamma_j$ fulfills this task.)
Now, consider
$$
  \rho'' = \frac{1}{2}\rho' + \frac{1}{2}U(F_j) \rho' U(F_j)^\dagger,
$$
for $j>1$.

Clearly, if $\rho'$ were a state, then the new operator $\rho''$ would also be a
state.
We claim that $\rho'$ has no terms with an index $j$ in its Clifford
basis expansion:
Note that if we flip the sign of precisely those
terms that have an index $j$ (i.e., they have a factor $\Gamma_j$
in the definition of the operator basis), and then the
coefficients cancel with those of $\rho'$. 

We now iterate this map through $j=2,3,\ldots, 2n$, and
we are left with a final state $\hat{\rho}$, which hence
must be of the form
\[
  \hat{\rho} = \frac{1}{d}\left( \1 + g_1'\Gamma_1 \right).
\]
By applying $U(T)^\dagger$ from above, we now transform $\hat{\rho}$ to
$U(T)^\dagger \hat{\rho} U(T) = \mathbb{P}(\rho)$, which is the first
part of the lemma.

Looking at $\hat{\rho}$ once more,
we see that this can be positive semidefinite only if $g_1' \leq 1$,
i.e., $\sum_{j=0}^{2n} g_j^2 \leq 1$.

Conversely, if $\sum_{j=0}^{2n} g_j^2 \leq 1$, then the (Hermitian)
operator $A = \sum_j g_j \Gamma_j$ has the property
\[
  A^2 = \sum_{jk} g_j g_k \Gamma_j\Gamma_k
      = \sum_j g_j^2 \1 \leq \1,
\]
i.e. $-\1 \leq A \leq \1$, so $\sigma = \frac{1}{d}(\1+A) \geq 0$.
\qed

\bigskip\noindent
{\bf Concavity of $\mathbf{f(t) = H\left( \frac{1 \pm \sqrt{t}}{2} \right)}$.}
Straightforward calculation shows that
\[
  f'(t) = \frac{1}{4\ln 2}\frac{1}{\sqrt{t}} \bigl( \ln(1-\sqrt{t})-\ln(1+\sqrt{t}) \bigr),
\]
and so
\[
  f''(t) = \frac{1}{8\ln 2}\frac{1}{t^{3/2}} \left( \ln\frac{1+\sqrt{t}}{1-\sqrt{t}}
                                                            - \frac{2\sqrt{t}}{1-t} \right).
\]
Since we are only interested in the sign of the second derivative, we ignore the
(positive) factors in front of the bracket, and are done if we can show that
\[\begin{split}
  g(t) &:= \ln \frac{1+\sqrt{t}}{1-\sqrt{t}} - \frac{2\sqrt{t}}{1-t} \\
       &=  \ln(1+\sqrt{t}) + \frac{1}{1+\sqrt{t}}
           - \ln(1-\sqrt{t}) - \frac{1}{1-\sqrt{t}}
\end{split}\]
is non-positive for $0 \leq t \leq 1$. Substituting $s = 1-\sqrt{t}$, which
is also between $0$ and $1$, we rewrite this as
\[
  h(s) = -\ln s - \frac{1}{s} + \ln(2-s) + \frac{1}{2-s},
\]
which has derivative
\[
  h'(s) = (1-s)\left( \frac{1}{s^2} - \frac{1}{(2-s)^2}\right),
\]
and this is clearly positive for $0<s<1$. In other words, $h$ increases
from its value at $s=0$ (where it is $h(0)=-\infty$) to its value at $s=1$
(where it is $h(1)=0$), so indeed $h(s) \leq 0$ for all $0\leq s\leq 1$.

Consequently, also $f''(t) \leq 0$ for $0\leq t \leq 1$, and we are done.
\qed

\bigskip\noindent
{\bf Constructive proof of Lemma 1.}
For the interested reader, we now give an explicit construction of the unitaries $U(T)$
and $U(F_j)$, which however requires a more intimate knowledge of the 
Clifford algebra. First of all, recall that we can write
two vectors $a,b \in \Real^{2n}$ in terms of the generators of the Clifford algebra
as $a = \sum_{j=1}^{2n} a_j \Gamma_j$ and $b = \sum_{j=1}^{2n} b_j \Gamma_j$. The Clifford product
of the two vectors is defined as $ab = a\cdot b + a \wedge b$, where $a \wedge b$ is the
outer product of the two vectors~\cite{lounesto:book,doran:lasenby:book}. When using 
the matrix representation of the Clifford algebra given above, this product is simply the
matrix product. Second, it is well known that within the Clifford algebra we may write 
the vector resulting from a reflection of the vector $a$ on the plane perpendicular to 
the vector $b$ (in 0) as $-b a b$. Rotations can then be expressed as 
successive reflections~\cite{lounesto:book,doran:lasenby:book}.

We first consider $U(T)$.
Here, our goal is to find the transformation $U(T)$ that rotates the vector $g = \sum_{j=0}^{2n} g_j \Gamma_j$
to the vector $b = \sqrt{\ell} \Gamma_1$, where we let $\ell := \sum_{j=0}^{2n} g_j^2$. Finding such a transformation
for only the first $2n$ generators can easily be achieved. The challenge is thus to include $\Gamma_0$. To this
end we perform three individual operations: First, we rotate $g' = \sum_{j=1}^{2n} g_j \Gamma_j$ onto the vector 
$b' = \sqrt{\ell'} \Gamma_1$ with $\ell' := \sum_{j=1}^{2n} g_j^2$. Second, we exchange $\Gamma_2$ and $\Gamma_0$. 
And finally we rotate the vector $g'' = \sqrt{\ell'} \Gamma_1 + g_0 \Gamma_2$ onto the vector $b = \sqrt{\ell}\Gamma_1$.

First,  we rotate $g' = \sum_{j=1}^{2n} g_j \Gamma_j$ onto the vector 
$b' = \sqrt{\ell'} \Gamma_1$:
Consider the vector
$\hat{g} = \frac{1}{\sqrt{\ell'}} g'$ .
We have $\hat{g}^2 = |\hat{g}|^2 \1 = \1$ and
thus the vector is of length 1. 
Let $m = \hat{g}+\Gamma_1$ denote the vector lying in the plane
spanned by $\Gamma_1$ and $\hat{g}$ located exactly half way between $\Gamma_1$ and $\hat{g}$. 
Let $\hat{m} = c (\hat{g}+\Gamma_1)$ with 
$c = \frac{1}{\sqrt{2(1+g_1/\sqrt{\ell'})}}$. It is easy to verify that $\hat{m}^2 = \1$ and hence the
vector $\hat{m}$ has length 1. To rotate the vector $g'$ onto the vector $b'$, we now need to first reflect
$g'$ around the plane perpendicular to $\hat{m}$, and then around the plane perpendicular to $\Gamma_1$.
Hence, we now define $R = \Gamma_1\hat{m}$. 
Evidently, $R$ is unitary since $RR^{\dagger} = R^{\dagger}R = \1$. 
First of all, note that
\begin{eqnarray*}
Rg' &=& \Gamma_1\hat{m}g'\\
& =& c \Gamma_1 \left(\frac{1}{\sqrt{\ell'}} g' + \Gamma_1\right)g' \\ 
&=& c \left(\Gamma_1\frac{a^2}{\sqrt{\ell'}} + \Gamma_1^2 g'\right)\\
&=& c \sqrt{\ell'}\left(\Gamma_1 + \frac{1}{\sqrt{\ell'}}g'\right)\\
&=&\sqrt{\ell'}\hat{m}.
\end{eqnarray*}
Hence,
$$
Rg'R^\dagger = \sqrt{\ell'}\hat{m}\hat{m}\Gamma_1 = \sqrt{\ell'}\Gamma_1 = b',
$$
as desired.
Using the geometry of the Clifford algebra, one can see that $k$-vectors remain $k$-vectors when transformed with the 
rotation $R$~\cite{doran:lasenby:book}. Similarly, it is easy to see that $\Gamma_0$ is untouched by the operation $R$
\begin{eqnarray*}
R \Gamma_0 R^\dagger = \Gamma_0 R R^\dagger = \Gamma_0,
\end{eqnarray*}
since $\{\Gamma_0,\Gamma_j\}=0$ for all $j \in \{1,\ldots,2n\}$.
We can thus conclude that 
\begin{eqnarray*}
R \rho R^\dagger = \frac{1}{d}\left(\1 + \sqrt{\ell'} \Gamma_1 + g_0 \Gamma_0 + \sum_{j<k} g'_{jk} \Gamma_{jk}+ \ldots\right),
\end{eqnarray*}
for some coefficients $g'_{jk}$.

Second, we exchange $\Gamma_2$ and $\Gamma_0$: To this end, recall that 
$\Gamma_2,\ldots,\Gamma_{2n},\Gamma_0$ is also a generating set for the Clifford algebra. Hence, we can now
view $\Gamma_0$ itself as a vector with respect to the new generators. 
To exchange $\Gamma_0$ and $\Gamma_2$, we now simply rotate $\Gamma_0$ onto $\Gamma_2$. Essentially,
this corresponds to a rotation about 90 degrees in the plane spanned by vectors $\Gamma_0$ and $\Gamma_2$. 
Consider the vector $n = \Gamma_0 + \Gamma_2$ located exactly in the middle between both vectors.
Let $\hat{n} = n/\sqrt{2}$ be the normalized vector. Let $R' = \Gamma_2 \hat{n}$. A small calculation anlogous to 
the above shows that 
$$
R' \Gamma_0 R{'^\dagger} = \Gamma_2\mbox{ and } R'\Gamma_2 R^{'\dagger} = - \Gamma_0.
$$
We also have that $\Gamma_1$, $\Gamma_3,\ldots,\Gamma_{2n}$ are untouched by the operation: for $j\neq 0$ 
and $j \neq 2$, we have that
$$
R' \Gamma_j R^{'\dagger} = \Gamma_j,
$$
since $\{\Gamma_0,\Gamma_j\} = \{\Gamma_2,\Gamma_j\} = 0$. How does $R'$ affect the $k$-vectors in terms
of the original generators $\Gamma_1,\ldots,\Gamma_{2n}$? Using the anti-commutation relations and the definition
of $\Gamma_0$ it is easy to convince yourself that all $k$-vectors are mapped to $k'$-vectors with $k' \geq 2$ (except for $\Gamma_0$ itself). Hence,
the coefficient of $\Gamma_1$ remains untouched.
We can thus conclude that 
\begin{eqnarray*}
R' R\rho R^\dagger R^{'\dagger} = \frac{1}{d}\left(\1 + \sqrt{\ell'} \Gamma_1 + g_0 \Gamma_2 + \sum_{j<k} g''_{jk} \Gamma_{jk} + \ldots\right),
\end{eqnarray*}
for some coefficients $g''_{jk}$.

Finally, we now rotate the vector $g'' = \sqrt{\ell'} \Gamma_1 + g_0 \Gamma_2$ onto the vector $b$.
Note that $(g'')^2 = (\ell + g_0^2) \1 = \ell \1$. Let $\hat{g}'' = g''/\sqrt{\ell}$ be the normalized
vector. Our rotation is derived exactly analogous to the first step: Let $k = \hat{g}'' + \Gamma_1$, and
let $\hat{k} = k/\sqrt{2(1+\sqrt{\ell'}/\sqrt{\ell})}$. Let $R'' = \Gamma_1\hat{k}$. A simple calculation analogous
to the above shows that
$$
R''g''R^{''\dagger} = \sqrt{\ell} \Gamma_1,
$$
as desired. Again, we have $R'' \Gamma_k R''^{\dagger} = \Gamma_k$ for $k \neq 1$ and $k \neq 2$. Furthermore,
$k$-vectors remain $k$-vectors under the actions of $R''$~\cite{doran:lasenby:book}.
Summarizing, we obtain
\begin{eqnarray*}
R'' R' R\rho R^\dagger R^{'\dagger} R^{''\dagger}= \frac{1}{d}\left(\1 + \sqrt{\ell} \Gamma_1 + \sum_{j<k} g'''_{jk} \Gamma_{jk}+ \ldots\right),
\end{eqnarray*}
for some coefficients $g'''_{jk}$. Thus, we can take $U(T) = R'' R' R$.

The argument for finding $U(F_j)$ is analogous. A simple computation using the fact that $\{\Gamma_0,\Gamma_j\} = 0$
for all $j$ gives us $U(F_j) = \Gamma_0 \Gamma_j$.
\end{document}